\documentclass[sigconf,anonymous=false]{acmart}
\usepackage{hyperref}
\usepackage{stfloats}
\usepackage{amsmath, amsfonts}

\usepackage{amssymb}

\usepackage{algorithmic}
\usepackage{graphicx}
\usepackage{textcomp}
\usepackage{xcolor}
\usepackage{algorithm}
\usepackage{float}
\usepackage{setspace}
\usepackage{subfigure}
\usepackage{verbatim}
\usepackage[misc]{ifsym}
\usepackage{enumitem}
\usepackage{caption}
\usepackage{extpfeil}
\usepackage{enumitem}

\AtBeginDocument{%
  \providecommand\BibTeX{{%
    \normalfont B\kern-0.5em{\scshape i\kern-0.25em b}\kern-0.8em\TeX}}}

\setcopyright{acmcopyright}
\copyrightyear{2018}
\acmYear{2018}
\acmDOI{10.1145/1122445.1122456}

\acmConference[Woodstock '18]{Woodstock '18: ACM Symposium on Neural
  Gaze Detection}{June 03--05, 2018}{Woodstock, NY}
\acmBooktitle{Woodstock '18: ACM Symposium on Neural Gaze Detection,
  June 03--05, 2018, Woodstock, NY}
\acmPrice{15.00}
\acmISBN{978-1-4503-XXXX-X/18/06}

\begin{document}

\title{Understanding the Interplay between Privacy and Robustness in Federated Learning}

\author{Yaowei Han}
\affiliation{%
  \institution{Department of Social Informatics, Kyoto University}
  \city{Kyoto}
  \country{Japan}}
\email{yaowei@db.soc.i.kyoto-u.ac.jp}

\author{Yang Cao}
\affiliation{%
  \institution{Department of Social Informatics, Kyoto University}
  \city{Kyoto}
  \country{Japan}}
\email{yang@i.kyoto-u.ac.jp}

\author{Masatoshi Yoshikawa}
\affiliation{%
  \institution{Department of Social Informatics, Kyoto University}
  \city{Kyoto}
  \country{Japan}}
\email{yoshikawa@i.kyoto-u.ac.jp}

\renewcommand{\shortauthors}{Han and Cao, et al.}

\begin{abstract}
Federated Learning (FL) is emerging as a promising paradigm of privacy-preserving machine learning, which trains an algorithm across multiple clients without exchanging their data samples.
Recent works highlighted several privacy and robustness weaknesses in FL and addressed these concerns using local differential privacy (LDP) and some well-studied methods used in conventional ML, separately.
However, it is still not clear how LDP affects adversarial robustness in FL.
To fill this gap, this work attempts to develop a comprehensive understanding of the effects of LDP on adversarial robustness in FL.
Clarifying the interplay is significant since this is the first step towards a principled design of private and robust FL systems.
We certify that local differential privacy has both positive and negative effects on adversarial robustness using theoretical analysis and empirical verification.
\end{abstract}






\maketitle

\section{Introduction}
Federated Learning (FL) \cite{mcmahan2017communication, kairouz2019advances, yang2019federated} is a relatively new and promising machine learning approach, emerging as a new paradigm of privacy-preserving machine learning.
It trains an algorithm across multiple clients (e.g., decentralized edge devices or servers holding local data samples) without exchanging their data samples.
In an FL system, data owners (participants) do not need to share raw data with the server.
Instead, participants jointly train an ML model by executing local training algorithms on their own private local data and only sharing model parameters with the parameter server.
This parameter server serves as a central aggregator to appropriately aggregate the local parameter updates and then share the aggregated updates with every participant.

While FL allows participants to keep their raw data local, recent works highlight and address the \textit{privacy} and \textit{robustness} concerns in FL.
To prevent potential privacy leakage from local parameter updates \cite{papernot2016semi}, local differential privacy (LDP) \cite{evfimievski2003limiting} has been adopted as a strong privacy guarantee in FL \cite{sun2020ldp,liuFedSelFederatedSGD2020, seifWirelessFederatedLearning2020a, truexLDPFedFederatedLearning2020} by locally adding perturbation to the updates.
Local differential privacy is a state-of-art paradigm facilitating secure analysis over sensitive data because of its strong assumption on adversary's background knowledge and smart setting of the privacy budget.
On the other hand, to avoid the adversarial examples \cite{yuanAdversarialExamplesAttacks2019, goodfellow2014explaining} and improve the adversarial robustness \cite{ dalvi2004adversarial}, researchers induce well-studied methods used in conventional ML such as bounding the norm of gradient updates or adding Gaussian noises \cite{sun2019can} for robust FL.

However, it is still not clear how LDP affects adversarial robustness in FL.
A few studies \cite{lecuyer2019certified, naseri2020toward} show that differential privacy may have positive effects on preventing adversarial examples in different settings.
Lecuyer et al. \cite{lecuyer2019certified} investigate a certified defense to adversarial examples in deep learning via a variant of differential privacy called PixelDP.
%to what robustness condition applying differential privacy can provide certified adversarial robustness with theoretical proof, but they did not consider the effect of local differential privacy on adversarial robustness generally.
%After noticing that randomization is the common building block for both (local) differential privacy and improving adversarial robustness, 
Naseri et al. \cite{naseri2020toward} present an empirical evaluation for the effect of (local) differential privacy on adversarial robustness in FL. 
%Although Naseri et al. \cite{naseri2020toward} put forward a detailed empirical plan, 
However, they do not have a theoretical basis to explain the effects of LDP on robustness in FL.
%The result shows that both DP and LDP benefit adversarial robustness under several attacks. 
%Certified robustness against adversarial examples proposed in \cite{lecuyer2019certified} further reveals the potential that differential privacy can benefit adversarial robustness.
More importantly, all existing studies only indicate that the higher level of perturbation on training samples could lead to higher adversarial robustness, which we call \textit{positive effect} in this work;
it is intriguing to study whether there is also a \textit{negative effect} of LDP on adversarial robustness, i.e., the robustness may decrease along with the higher level of privacy.

To fill this gap, we attempt to develop a comprehensive understanding of the effects of local differential privacy on adversarial robustness in federated learning.
%Adversarial robustness is a popular topic and is commonly used for theoretical analysis in the machine learning area.
Specifically, we will show that a higher level of LDP (lower $\epsilon$) may not always lead to a higher level of adversarial robustness.
Given the increasing attention of LDP-based FL \cite{sun2020ldp,liuFedSelFederatedSGD2020, seifWirelessFederatedLearning2020a, truexLDPFedFederatedLearning2020},
clarifying the interplay between LDP and adversarial robustness in FL is significant since this is the first step towards a principled design of private and robust FL systems.  %instructs the usage of local differential privacy for the practice of federated learning.

%We tend to certify that local differential privacy has both positive and negative effects on adversarial robustness using theoretical proof and empirical verification. 

Our contributions in this work are three-fold.
\begin{itemize} [leftmargin=*, topsep=0pt]
\item To fill in the theoretical gap that whether local differential privacy benefits adversarial robustness, we show the theoretical evidence about the connection between local differential privacy and our proposed adversarial robustness called {$E_{e^\epsilon}$ - adversarial robustness} (Section 3).
\item We then study the mixed effects of local differential privacy on preventing adversarial examples by clarifying the connection between our proposed robustness definition and the ones defined in the literature.
We qualitatively show that LDP may also bring a negative effect on robustness, i.e., a higher level of privacy (lower $\epsilon$) can lead to lower adversarial robustness (Section 4).
\item We conduct experiments to verify how different privacy parameters of LDP affect the adversarial robustness in FL, which is in line with our theoretical analysis (Section 5).
\end{itemize}

\section{PROBLEM SETTING} \label{sec:problem_setting}

This work wants to learn about the effect of local differential privacy on adversarial robustness in federated learning.
Fig.\ref{fig:problem} shows the federated learning scenario we focus on.
The participants and the server want to train a model collaboratively.
In the training phase, due to some privacy concerns, the participants will download the global model, compute the local model, update local updates.
We suppose that each participant obeys the rule that adding local differential privacy using the same $\epsilon$ to its local data.
In the inference phase, the adversary tries to use adversarial examples to attack the model. We constrain the adversarial examples to be norm-bounded.

Overall, our work aims to

\begin{enumerate}
    \item Theoretically dig the theoretical evidence about why LDP brings the positive effect.
    \item Theoretically find the overall effect of local differential privacy on adversarial robustness.
    \item Empirically see the change of adversarial robustness on different settings of $\epsilon$.
\end{enumerate}
 
Here we consider a classification task in federated learning with data $(x, y) \in \mathcal{X} \times \mathcal{Y}(\mathcal{Y} = \{1,...,C\}$) from a distribution D.
We assume the existence
of a \textit{labeling oracle} $\mathcal{O} : \mathcal{X} \rightarrow \mathcal{Y} \cup \{\bot\}$ that maps any input in $\mathcal{X}$ to its true label, or to the “un-labelable”. Take a digit classification task for example, the oracle $\mathcal{O}$ corresponds to human labeling of any image as a digit.
Note that for $(x, y) \sim D$, we always have $y = \mathcal{O}(x)$.
The goal is to learn a classifier $f$: $\mathcal{X} \rightarrow \mathcal{Y}$ that agrees with the oracle’s labels.

\begin{figure}[tb]
  \centering
  \includegraphics[width=0.45\textwidth]{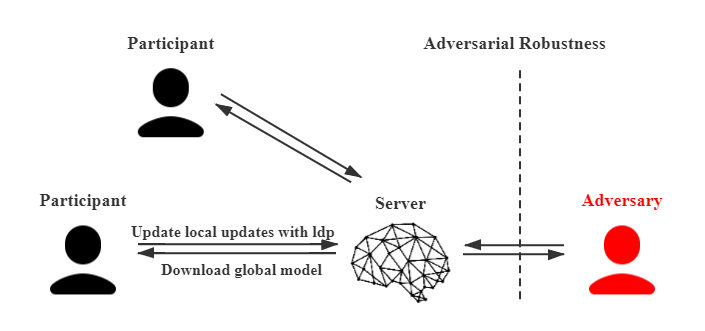}
  \caption{Problem setting.}
  \label{fig:problem}
\end{figure}

\section{Theoretical Evidence of Positive Effect of LDP} \label{sec:evi}

To find the theoretical evidence of why LDP brings a positive effect to adversarial robustness, we first find the connection between local differential privacy definition and adversarial robustness.
We notice that randomization is the common keyword of both local differential privacy and adversarial robustness.
While the core technology of local differential privacy is to inject noises based on the definition, randomization has also proven effective in adversarial defense strategies \cite{liu2018towards}.
The technologies used are to inject random noise (both during training and inference phases) inside the machine learning network architecture,i.e.at a given layer of the network.
Thus, the intuition is that there would be a connection between noises we add in local differential privacy and randomization we apply for generating the adversarial examples.

We find that \cite{pinot2019unified} highlights some links between Renyi differential privacy \cite{mironov2017renyi} and generalized adversarial robustness \cite{pinot2019unified}.
This presents a unified view towards both but doesn't further discuss the practical meaning. It only cares about the generalized differential privacy and generalized adversarial robustness which cannot be used when we discuss the effect of LDP.
We then want to find a type of adversarial robustness definition, which we can see as a unified view of local differential privacy.
We first give the definition of local differential privacy.

\begin{definition}
Local differential privacy \cite{evfimievski2003limiting}: A randomized algorithm (mechanism) $\mathcal{M}: \mathcal{X} \rightarrow \mathcal{Y}$ is $\epsilon$-locally differentially private if for any $x, x' \in \mathcal{X}$ and $v \in \mathcal{Y}$, we have
\begin{equation}
Pr[\mathcal{M}(x) = v] \leq e^\epsilon Pr[\mathcal{M}(x') = v]
\end{equation}
where lower $\epsilon$ means higher privacy degree.
\end{definition}

According to the definition of local differential privacy, we propose a self-defined adversarial robustness named $E_{e^\epsilon}$ - adversarial robustness.

\begin{definition}
$E_{e^\epsilon}$ - adversarial robustness. A randomized classifier M is said to be $E_{e^\epsilon}$ - robust if $P_{x\sim D_\mathcal{X}}[\exists x' \in B(x, \alpha)$ $s.t$ $E_\lambda(\mathcal{M}(x'),$ $\mathcal{M}(x))$ $> 0] = 0$ where $\lambda = e^\epsilon$ and $B(x, \alpha) = \{x' \in \mathcal{X}$ $s.t$ $d_\mathcal{X}(x, x') \leq \alpha$, $\alpha = \infty\}$.
\end{definition}

Then we can see the connection between local differential privacy and $E_{e^\epsilon}$ - adversarial robustness.

\begin{theorem}
An algorithm $\mathcal{M}$ is $E_{e^\epsilon}$ - robust if and only if $\mathcal{M}$ is $\epsilon$ - local differential private.
\end{theorem}

\begin{proof}
We can rewrite the definition of local differential privacy as follows.

Let $\epsilon > 0$, $(\mathcal{X}, d_\mathcal{X})$ an arbitrary (input) metric space, and $\mathcal{Y}$ the output space. A probabilistic mapping $\mathcal{M}$ from $\mathcal{X}$ to $\mathcal{Y}$ is called $\epsilon$ - local differential private if for any $x, x'$, one has 
\begin{equation}
E_\lambda(\mathcal{M}(x), \mathcal{M}(x'))  = 0   
\end{equation}
where $\lambda = e^\epsilon$. $E_\lambda$ here is $E_\lambda$-divergence (hockey-stick divergence) which is a special type of $f$- divergence with $f_\lambda(t)=(t-\lambda)^+$.

We can prove the definition is a equipment of the original definition by:
\begin{equation}
\begin{split}
E_{e^\epsilon}(\mathcal{M}(x), \mathcal{M}(x')) &= 0 \\
(\frac{Pr[\mathcal{M}(x) = v]}{Pr[\mathcal{M}(x') = v]} - e^\epsilon)^+ &= 0 (v \in \mathcal{Y})\\
\frac{Pr[\mathcal{M}(x) = v]}{Pr[\mathcal{M}(x') = v]} &\leq e^\epsilon
\end{split}
\end{equation}

Comparing the definition of above local differential privacy and $E_{e^\epsilon}$ - adversarial robustness, we can get the conclusion that an algorithm $\mathcal{M}$ is $E_{e^\epsilon}$ - robust if and only if $\mathcal{M}$ is $\epsilon$ - local differential private.
\end{proof}

This conclusion indicates that applying local differential privacy improves adversarial robustness to some extent.

\section{THE EFFECT OF LOCAL DIFFERENTIAL PRIVACY ON ADVERSARIAL ROBUSTNESS} \label{sec:effect}

Section \ref{sec:evi} already shows that local differential privacy and $E_{e^\epsilon}$ - adversarial robustness are equivalent.
However, there is a gap between $E_{e^\epsilon}$ - adversarial robustness and adversarial robustness in the literature.
This section will firstly clarify the comprehensive definition of norm-bounded adversarial robustness we focus on and then qualitatively study how local differential privacy affects the norm-bounded adversarial robustness.
We then find that besides positive effects, local differential privacy can also bring negative effects on norm-bounded adversarial robustness.

\subsection{Definition of Adversarial Robustness}

The definition of adversarial robustness depends on the method of generating adversarial examples.
At its broadest, the definition of an adversarial example consists in any adversarial failure induced in a classifier \cite{goodfellow2014explaining}.
Tramèr et al. \cite{tramer2020fundamental} summary two types of adversarial examples: sensitivity adversarial examples and invariance adversarial examples.

Our work focuses on adversarial robustness defined with norm-bounded adversarial examples, which are generally used for evaluating adversarial robustness.
Norm-bounded adversarial examples are to constrain the amount of change an attacker is allowed to make to the input. The change of input is measured by the $p$-norm of the change denoted by $\lVert x^* - x \rVert_p$. Thus, we get the definition of both norm-bounded sensitivity adversarial examples and norm-bounded invariance adversarial examples.

%% clarify what kind of robust
\textbf{Definition 1} (Norm-bounded Sensitivity Adversarial Examples).
Given a classifier $f$ and a correctly classified input $(x, y) \sim D$ (i.e.,
$\mathcal{O}(x) = f(x) = y$), an "$\alpha$-bounded sensitivity adversarial
example is an input $x^* \in \mathcal{X}$ such that:
\begin{enumerate}
    \item $f(x^*) \neq f(x)$.
    \item $\lVert x^* - x \rVert_p \leq \alpha$.
\end{enumerate}

\textbf{Definition 2} (Norm-bounded Invariance Adversarial Examples).
Given a classifier $f$ and a correctly classified input $(x, y) \sim D$ (i.e.,
$\mathcal{O}(x) = f(x) = y$), an "$\alpha$-bounded invariance adversarial
example is an input $x^* \in \mathcal{X}$ such that:
\begin{enumerate}
    \item $f(x^*) = f(x)$.
    \item $\mathcal{O}(x^*) \neq \mathcal{O}(x)$, and $\mathcal{O}(x^*) != \bot$
    \item $\lVert x^* - x \rVert_p \leq \alpha$.
\end{enumerate}

Thus, we say that a model $f$ is robust to $p$-norm attacks on a given input $x$ if for all $x^* s.t. \lVert x^* - x \rVert_p \leq \alpha$, we have
\begin{enumerate}
    \item $f(x) = \mathcal{O}(x) = y$.
    \item $f(x^*) = \mathcal{O}(x^*) = y$.
\end{enumerate}

\subsection{The Effect of Local Differential Privacy on Adversarial Robustness} \label{sec:eff}

In this section, we tend to qualitatively and overall study how local differential privacy affects the norm-bounded adversarial robustness we focus on.

Combining with the result of the section \ref{sec:evi}, we can see that local differential privacy is equivalent to $E_{e^\epsilon}$ - adversarial robustness using a unified view. We then consider the exact effect of local differential privacy on norm-bounded adversarial robustness. That is, we will discuss the connection and difference between $E_{e^\epsilon}$ - adversarial robustness and norm-bounded adversarial robustness towards sensitivity adversarial examples and invariance adversarial examples.

Observing the definition of $E_{e^\epsilon}$ - adversarial robustness, we can see the biggest difference with norm-bounded adversarial robustness is the value of $\alpha$. Norm-bounded adversarial robustness uses $\alpha$ to constrain the change of input. $E_{e^\epsilon}$ - adversarial robustness here instead uses the parameter $\epsilon$ to do the similar work. We can see that $\epsilon$ is also to bound the difference of $\mathcal{M}(x)$ and $\mathcal{M}(x')$.
That is, the smaller the $\epsilon$ is, the smaller the distribution divergence of inputs is.
We then seek how the $\epsilon$ here in the $E_{e^\epsilon}$ - adversarial robustness affects the norm-bounded adversarial robustness.
Note that $\epsilon$ is actually induced from the definition from LDP, we return to the mechanism of local differential privacy. We usually add some noises to make our mechanism satisfy the definition. The extent of noising (roughly the same with the $\alpha$ we use in adversarial examples) is in inverse proportion to the value of $\epsilon$. In conclusion, we have:

\begin{quote}
    The effect of $\epsilon$ in $E_{e^\epsilon}$ - adversarial robustness is inverse proportion to the effect of $\alpha$ in norm-bounded adversarial robustness.
\end{quote}

\begin{figure}[tb]
  \centering
  \includegraphics[width=0.5\textwidth]{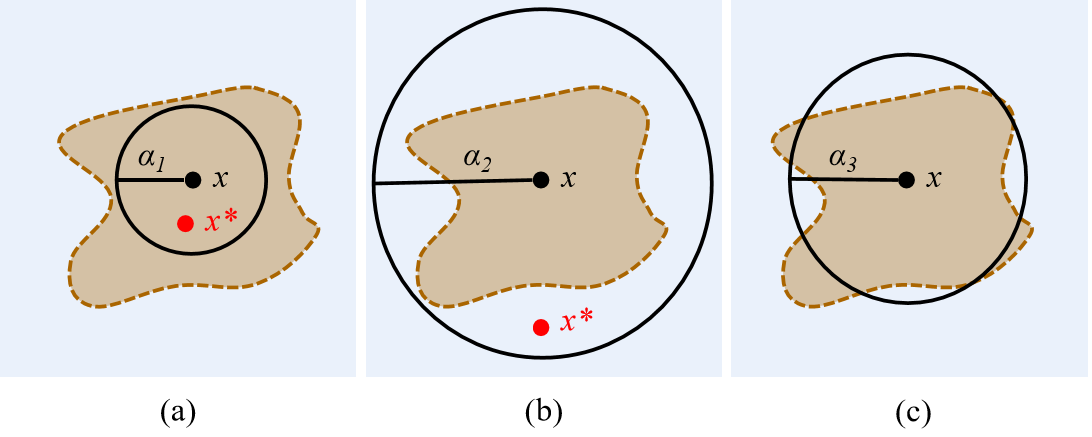}
  \caption{Generate adversarial examples with different settings of p-norm ball radius.}
  \label{fig:adversarial}
\end{figure}

For knowing the effect of $\alpha$ in norm-bounded adversarial robustness, we can mock the generation of adversarial examples.
Fig.\ref{fig:adversarial} shows three possibilities. The input space is (ground-truth) classified into the solid brown region and the blue region. The brown dotted line is the decision boundary. Given a data point (the central black point) $x$, we know its true classification and want to use it to generate some adversarial examples. Due to the norm used to define “small” adversarial perturbations is misaligned with the labeling oracle $\mathcal{O}$, we have (as show in Fig.\ref{fig:adversarial} - a, b):
\begin{equation*}
    \alpha_1 := \min{\lVert \Delta \rVert : \mathcal{O}(x + \Delta) \notin \{y, \bot\} }
\end{equation*}
\begin{equation*}
    \alpha_2 := \max{\lVert \Delta \rVert : \mathcal{O}(x + \Delta) \in \{y, \bot\} }
\end{equation*}
Then the adversarial examples with $x^*$ $s.t. \lVert x^* - x \rVert \leq \alpha_1$ we generate broaden the oracle area and can be used to strengthen the adversarial robustness. Thus, with $\alpha$ $s.t.$ $0 \leq \alpha \leq \alpha_1$ increases, the bigger oracle area we broaden, the higher adversarial robustness. In extreme cases with $\alpha = 0$, the solid brown area except $x$ can be a potential successful attack. Similarly, the adversarial examples with $x^*$ $s.t. \lVert x^* - x \rVert \geq \alpha_2$ we generate induces overly - robustness due to the presence of invariance adversarial examples \cite{tramer2020fundamental}. Thus, with $\alpha$ $s.t.$ $\alpha \geq \alpha_2$ increases, the lower adversarial robustness. In extreme cases with $\alpha = \infty$ which is definitely impractical, the whole dataset will be corrupted.
Intuitively, there should be a balance area that can benefit the adversarial robustness the best when we choose radiuses like $\alpha_3$ $s.t.$ $\alpha_1 \leq \alpha_3 \leq \alpha_2$.
To be concluded, as the radius of $p$-norm ball $\alpha$ increases, the adversarial robustness will first increase, access the balance area, and then decrease.

Based on the above relationship between $\alpha$ and $\epsilon$ we can have the conclusion which is:

\begin{quote}
    While privacy degree increases with $\epsilon$ decreasing, adversarial robustness firstly increases, access the balance area, then decreases.
\end{quote}

\section{EXPERIMENTS AND EVALUATION} \label{sec:exp}

Section \ref{sec:effect} gives the intuition about the relationship between local differential privacy and adversarial robustness.
To verify this intuition, we implement an experimental test to see the adversarial robustness in different privacy budget settings.

\subsection{Experimental Setup}

We experiment with different settings of $\epsilon$ in FL and observe the change of norm-bounded adversarial robustness on different $\epsilon$.

We use two datasets for our experiments:
1) MNIST \cite{lecun1998gradient}, to ease comparisons, and 2) CIFAR10 \cite{hope2017learning}, to extend the representativeness of our evaluation.
We use the lightweight CNN model \cite{he2016deep} for training.
All experiments use PyTorch \cite{paszke2017automatic}.
Our source code is available in Github \footnote{https://github.com/iris0305/privacy\_robust}.

\subsection{Adversarial Robustness Measurements to Adversarial Example}

Goodfellow et al. \cite{goodfellow2014explaining} proposed an efficient single iteration method, using backpropagation, to compute an $l_{\infty}$ bounded adversarial perturbation for a given input $x$ called Fast Gradient Sign Method (FGSM):

\begin{equation}
    x_{adv} = x + \alpha sign({\nabla}_x\mathcal{L}(\theta, x, y))
\end{equation}

where $y$ is the true class of input $x$, $\theta$ is the model parameters, ${\nabla}_x\mathcal{L}(\theta, x, y)$ computes the gradient of a cost function with respect
to $x$. FGSM serves as a simple yet effective way of testing the robustness of a neural network.

From \cite{yu2019interpreting}, we know that models are not always robust -- the predictions of a model on clean and noisy inputs are not always the same and can diverge to a large extent with small adversarial noises. As such, given a feasible perturbation set, the average of Kullback–Leibler divergences between the original predictions and the adversarial predictions could be used to denote the model’s vulnerability degree (i.e., the inverses of model robustness). Thus, we can calculate the average Kullback–Leibler divergence between two predictions on original inputs and adversarial inputs with perturbations in a defined range. The formal robustness could be estimated by:

\begin{equation}
    \psi(x) = \frac{1}{{\rm avg} \,  D_{KL}(f(x), f(x_{adv}))}
\end{equation}

\subsection{Experiment Results} \label{sec:expe1}
%% TODO

We apply several settings of $\alpha$ for measuring the norm-bounded adversarial robustness. We use the average Kullback–Leibler divergence between two predictions on original inputs and adversarial inputs with perturbations to measure the norm-bounded adversarial robustness. For each setting of $\alpha$, we run experiments five times and calculate the average for the reason that local differential privacy is based on randomization.
Fig.\ref{fig:expe1} and Fig.\ref{fig:expe2} show the results when the dataset are MNIST and CIFAR10, respectively. 

From Fig.\ref{fig:expe1} and Fig.\ref{fig:expe2}, we can see that both results for MNIST and CIFAR10 agree with the expectation in section \ref{sec:eff}. While privacy degree decreases with epsilon $\epsilon$ (x-axis) increasing, adversarial robustness evaluated based on KL divergence (y-axis) firstly increases, access the balance area, then decreases. Note that in the middle range, we have a balance area since there induces mixed "good" and "bad" adversarial examples as showed in Fig.\ref{fig:adversarial} - c.

\begin{figure}[tb]
  \centering
  \includegraphics[width=0.5\textwidth]{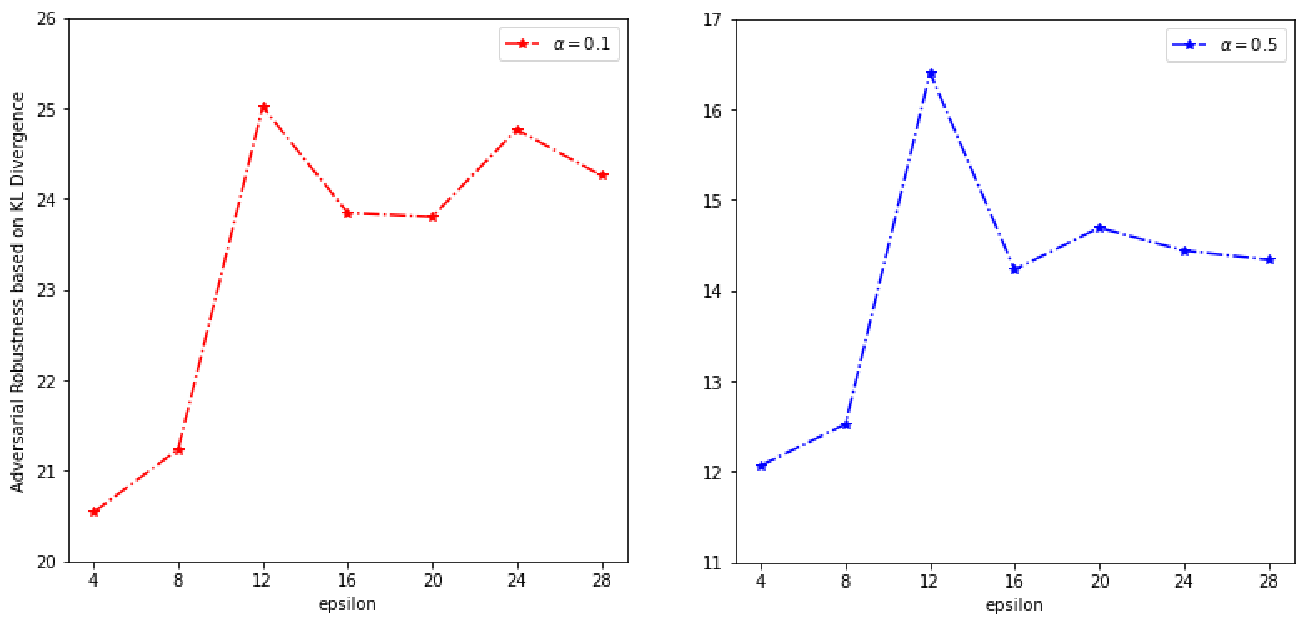}
  \caption{Experiment results for MNIST.}
  \label{fig:expe1}
\end{figure}

\begin{figure}[tb]
  \centering
  \includegraphics[width=0.5\textwidth]{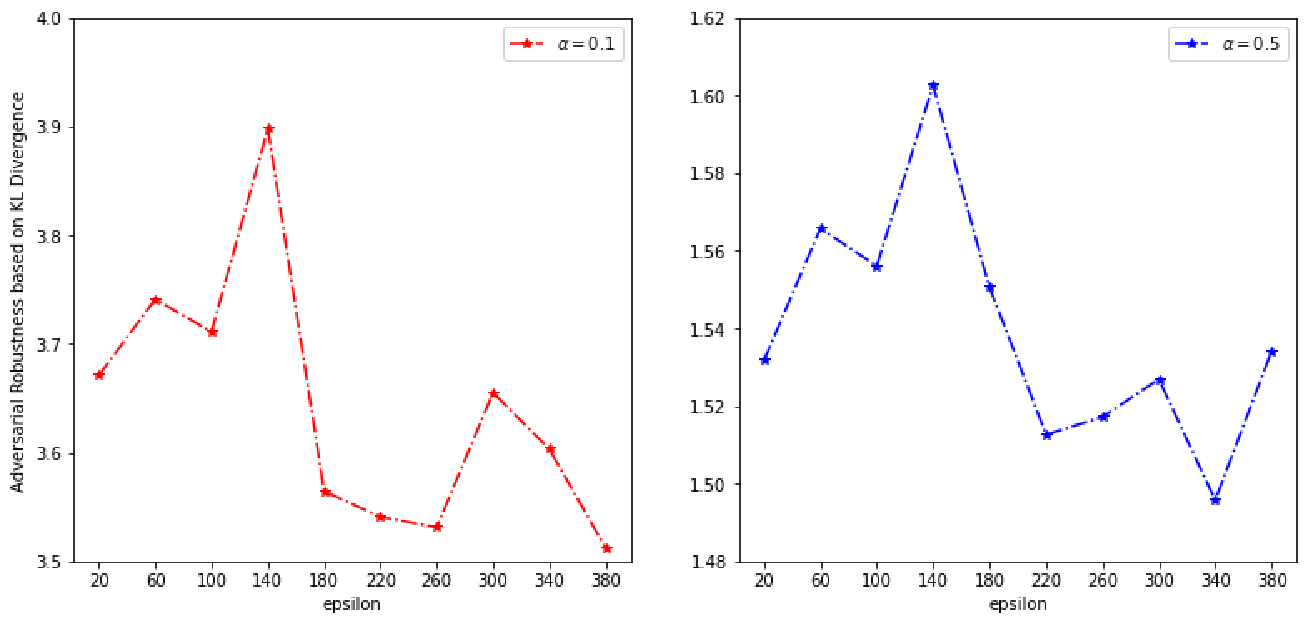}
  \caption{Experiment results for CIFAR10.}
  \label{fig:expe2}
\end{figure}
%% TODO

\section{CONCLUSION} \label{sec:Future Work}

This work uses theoretical analysis and empirical proof to study the exact effects of local differential privacy on adversarial robustness in federated learning. 
We find that local differential privacy has the potential for both positive and negative effects, which means that a higher level of local differential privacy (lower $\epsilon$) can bring higher or lower adversarial robustness.
For future work, we will further put forward more theoretical proof. We will also make a more comprehensive empirical plan.

\bibliographystyle{ACM-Reference-Format}
\bibliography{ref}


\end{document}